\documentclass[conference,letterpaper,romanappendices]{ieeeconf}
\let\proof\relax   

\usepackage{amsthm,xpatch}
\usepackage{amsmath,amsfonts}
\usepackage{cite}
\usepackage{amssymb}
\usepackage{dsfont}
\usepackage{graphicx, subfigure}
\usepackage{color}
\usepackage{breqn}
\usepackage{mathtools}
\usepackage{bbm}
\usepackage{latexsym}
\usepackage[ruled, linesnumbered]{algorithm2e}
\usepackage{accents}
\usepackage{tikz}

\newtheorem{lemma}{Lemma}
\newtheorem{theorem}{Theorem}

\newcommand*{\transpose}{%
  {\mathpalette\@transpose{}}%
}

\IEEEoverridecommandlockouts

\begin{document}

\newcommand{\SB}[3]{
\sum_{#2 \in #1}\biggl|\overline{X}_{#2}\biggr| #3
\biggl|\bigcap_{#2 \notin #1}\overline{X}_{#2}\biggr|
}

\newcommand{\Mod}[1]{\ (\textup{mod}\ #1)}

\newcommand{\overbar}[1]{\mkern 0mu\overline{\mkern-0mu#1\mkern-8.5mu}\mkern 6mu}

\makeatletter
\newcommand*\nss[3]{%
  \begingroup
  \setbox0\hbox{$\m@th\scriptstyle\cramped{#2}$}%
  \setbox2\hbox{$\m@th\scriptstyle#3$}%
  \dimen@=\fontdimen8\textfont3
  \multiply\dimen@ by 4             
  \advance \dimen@ by \ht0
  \advance \dimen@ by -\fontdimen17\textfont2
  \@tempdima=\fontdimen5\textfont2  
  \multiply\@tempdima by 4
  \divide  \@tempdima by 5          
  \ifdim\dimen@<\@tempdima
    \ht0=0pt                        
    \@tempdima=\fontdimen5\textfont2
    \divide\@tempdima by 4          
    \advance \dimen@ by -\@tempdima 
    \ifdim\dimen@>0pt
      \@tempdima=\dp2
      \advance\@tempdima by \dimen@
      \dp2=\@tempdima
    \fi
  \fi
  #1_{\box0}^{\box2}%
  \endgroup
  }
\makeatother

\makeatletter
\renewenvironment{proof}[1][\proofname]{\par
  \pushQED{\qed}%
  \normalfont \topsep6\p@\@plus6\p@\relax
  \trivlist
  \item[\hskip\labelsep
        \itshape
    #1\@addpunct{:}]\ignorespaces
}{%
  \popQED\endtrivlist\@endpefalse
}
\makeatother

\makeatletter
\newsavebox\myboxA
\newsavebox\myboxB
\newlength\mylenA

\newcommand*\xoverline[2][0.75]{%
    \sbox{\myboxA}{$\m@th#2$}%
    \setbox\myboxB\null
    \ht\myboxB=\ht\myboxA%
    \dp\myboxB=\dp\myboxA%
    \wd\myboxB=#1\wd\myboxA
    \sbox\myboxB{$\m@th\overline{\copy\myboxB}$}
    \setlength\mylenA{\the\wd\myboxA}
    \addtolength\mylenA{-\the\wd\myboxB}%
    \ifdim\wd\myboxB<\wd\myboxA%
       \rlap{\hskip 0.5\mylenA\usebox\myboxB}{\usebox\myboxA}%
    \else
        \hskip -0.5\mylenA\rlap{\usebox\myboxA}{\hskip 0.5\mylenA\usebox\myboxB}%
    \fi}
\makeatother

\xpatchcmd{\proof}{\hskip\labelsep}{\hskip3.75\labelsep}{}{}

\pagestyle{plain}

\title{\fontsize{21}{28}\selectfont Capacity of Single-Server Single-Message Private Information Retrieval with Coded Side Information}

\author{Anoosheh Heidarzadeh, Fatemeh Kazemi, and Alex Sprintson\thanks{The authors are with the Department of Electrical and Computer Engineering, Texas A\&M University, College Station, TX 77843 USA (E-mail: \{anoosheh, fatemeh.kazemi, spalex\}@tamu.edu).}}


\maketitle 

\thispagestyle{plain}

\begin{abstract} 
This paper considers the problem of single-server single-message private information retrieval with coded side information (PIR-CSI). In this problem, there is a server storing a database, and a user which knows a linear combination of a subset of messages in the database as a side information. The number of messages contributing to the side information is known to the server, but the indices and the coefficients of these messages are unknown to the server. The user wishes to download a message from the server privately, i.e., without revealing which message it is requesting, while minimizing the download cost. In this work, we consider two different settings for the PIR-CSI problem depending on the demanded message being or not being one of the messages contributing to the side information. For each setting, we prove an upper bound on the maximum download rate as a function of the size of the database and the size of the side information, and propose a protocol that achieves the rate upper-bound.
\end{abstract}

\section{introduction}
In the original setting of the private information retrieval (PIR) problem~\cite{Chor:PIR1995}, a user wishes to download (with minimum cost) a message belonging to a database of $K$ messages privately, i.e., without revealing which message it is requesting, from a single server or multiple servers each storing a copy of the database. In a single-server setting or a multiple-server setting when the servers collude, in order to achieve privacy in an information-theoretic sense, the user must download the whole database~\cite{Chor:PIR1995}. However, when the database is replicated on multiple non-colluding servers (see, e.g.,~\cite{Sun2017, JafarPIR3}), or coded versions of the data are stored on the servers (e.g., see~\cite{BU18, tajeddine2017private2}), or the user has some side information about the database (see, e.g.,~\cite{Kadhe2017,Tandon2017,Wei2017FundamentalLO,Wei2017CacheAidedPI,Chen2017TheCO}), the user can achieve the information-theoretic privacy more efficiently than downloading the whole database. The multi-message setting of PIR problem has also been studied in~\cite{BU17,Maddah2018}, where the user wishes to download multiple messages privately, instead of only one message as in the single-message setting, from a single server or multiple servers.  

In this paper, we study the single-server single-message PIR problem when the user knows a linear combination of a subset of $M$ messages in the database as a side information. This problem generalizes those previously studied in the single-server single-message PIR setting. In particular, we assume that the indices and coefficients of the messages contributing to the user's side information are unknown to the server, and the user's demanded message may or may not be one of the messages in the side information. This type of side information can be motivated by several scenarios. For example, the user may have overheard some coded packets over a wireless broadcast channel, or some part of the user's information, which is locally stored using an erasure code, may be lost and not recoverable locally. 

\subsection{Main Contributions}
In this work, we characterize the capacity of the PIR-CSI problem, defined as the supremum of all achievable
download rates, in a single-server single-message setting as a function of the size of the database ($K$) and the size of the side information ($M$). In particular, for the setting in which the user's demand is not one of the messages contributing to its side information, we prove that the capacity is equal to $\lceil\frac{K}{M+1}\rceil^{-1}$ for any $0\leq M< K$. Interestingly, the capacity of PIR with (uncoded) side information~\cite{Kadhe2017} is also equal to $\lceil\frac{K}{M+1}\rceil^{-1}$ where $M$ is the number of messages available at the user. This shows that there will be no loss in capacity, when compared to the case that the user knows $M$ messages separately, even if the user knows only \emph{one} linear combination of $M$ messages. Also, for the setting in which the demanded message is contributing to the user's side information, we prove that the capacity is equal to~$1$ for $M=2$ and $M=K$, and is equal to~$\frac{1}{2}$ for any ${3\leq M\leq K-1}$. This is interesting because it shows that, no matter what the size of the side information is, the user can privately retrieve any message contributing to its side information with a download cost at most twice the cost of downloading the message directly. The proof of converse for each setting is based on information-theoretic arguments, and for the achievability proofs, for different cases of $M$, we propose different PIR protocols which are all based on the idea of randomized non-uniform partitioning.

\section{Problem Formulation}\label{sec:SN}
Let $q\geq 3$, $m\geq 1$, $K\geq 1$, and $0\leq M\leq K$ be integers. Let $\mathbb{F}_q$ be a finite field of size $q$, and let $\mathbb{F}_{q^m}$ be an extension field of $\mathbb{F}_q$ of size $q^m$. Let $\mathbb{F}^{\times}_q$ be the multiplicative group of $\mathbb{F}_q$, i.e., $\mathbb{F}_q^{\times} \triangleq \mathbb{F}_q\setminus \{0\}$. For a positive integer $i$, denote $[i]\triangleq\{1,\dots,i\}$, and let $[0]\triangleq\emptyset$. 

We assume that there is a server storing a set $X$ of $K$ messages $X_1,\dots,X_K$, with each message $X_i$ being independently and uniformly distributed over $\mathbb{F}_{q^m}$, i.e., ${H(X_1) = \dots = H(X_K) = L}$ and $H(X_1,\dots,X_K) = KL$, where $L \triangleq m\log_2 q$. Also, we assume that there is a user that wishes to retrieve a message $X_W$ from the server for some $W\in [K]$, and knows a linear combination ${Y^{[S,C]}\triangleq \sum_{i\in S} c_i X_i}$ for some $S \triangleq \{i_1,\dots,i_M\}\in \mathcal{S}$ and some ${C \triangleq \{c_{i_1},\dots,c_{i_M}\} \in \mathcal{C}}$, where $\mathcal{S}$ is the set of all subsets of $[K]$ of size $M$, and $\mathcal{C}$ is the set of all ordered sets of size $M$ (i.e., all sequences of length $M$) with elements from $\mathbb{F}^{\times}_q$. We refer to $W$ as the \emph{demand index}, $X_W$ as the \emph{demand}, $Y^{[S,C]}$ as the \emph{side information}, and $M$ as the \emph{side information size}. 


Let $\boldsymbol{S}$, $\boldsymbol{C}$, and $\boldsymbol{W}$ be random variables representing $S$, $C$, and $W$, respectively. Denote the probability mass function (pmf) of $\boldsymbol{S}$ by $p_{\boldsymbol{S}}(.)$, the pmf of $\boldsymbol{C}$ by $p_{\boldsymbol{C}}(.)$, and the conditional pmf of $\boldsymbol{W}$ given $\boldsymbol{S}$ by $p_{\boldsymbol{W}|\boldsymbol{S}}(.|.)$. 

We assume that $\boldsymbol{S}$ is uniformly distributed over $\mathcal{S}$, i.e., 
\[p_{\boldsymbol{S}}(S) = \frac{1}{\binom{K}{M}}, \quad S\in \mathcal{S},\] and $\boldsymbol{C}$ is uniformly distributed over $\mathcal{C}$, i.e., 
\[p_{\boldsymbol{C}}(C) = \frac{1}{(q-1)^{M}}, \quad C\in\mathcal{C}.\] Also, we consider two different models for the conditional pmf of $\boldsymbol{W}$ given $\boldsymbol{S}=S$ as follows: 

\subsubsection*{Model~I} $\boldsymbol{W}$ is uniformly distributed over $[K]\setminus S$, i.e., 
\begin{equation*}
p_{\boldsymbol{W}|\boldsymbol{S}}(W|S) = 
\left\{\begin{array}{ll}
\frac{1}{K-M}, & W\not\in S,\\	
0, & \text{otherwise}.
\end{array}\right.	
\end{equation*} 

\subsubsection*{Model~II} $\boldsymbol{W}$ is uniformly distributed over $S$, i.e., 
\begin{equation*}
p_{\boldsymbol{W}|\boldsymbol{S}}(W|S) = 
\left\{\begin{array}{ll}
\frac{1}{M}, & W\in S,\\	
0, & \text{otherwise};
\end{array}\right.	
\end{equation*} (Note that the model~I is valid for $0\leq M<K$, and the model~II is valid for $0< M\leq K$.) 

Let $I^{[W,S]}$ be an indicator function such that $I^{[W,S]} = 1$ if $W\in S$, and $I^{[W,S]} = 0$ if $W\not\in S$. In the model~I, ${I^{[W,S]} = 0}$, and in the model~II, ${I^{[W,S]}=1}$. 

We assume that $I^{[W,S]}$ is known to the server a priori. We also assume that the server knows the size of $S$ (i.e., $M$) and the pmf's $p_{\boldsymbol{S}}(.)$, $p_{\boldsymbol{C}}(.)$, and $p_{\boldsymbol{W}|\boldsymbol{S}}(.|.)$, whereas the realizations $S$, $C$, and $W$ are unknown to the server a priori.

For any $S$, $C$, and $W$, in order to retrieve $X_W$, the user sends to the server a query $Q^{[W,S,C]}$, which is a (potentially stochastic) function of $W$, $S$, $C$, and $Y$, and is independent of any $Y'=\sum_{i\in S'} c'_i X_i$ where $S'\subseteq [K]$ and $C'=\{c'_1,\dots,c'_{|S'|}\}$, $c'_i\in\mathbb{F}^{\times}_q$ such that $(S',C')\neq (S,C)$, i.e., \[I(Y'; Q^{[W,S,C]})=0.\]

The query $Q^{[W,S,C]}$ must protect the privacy of the user's demand index $W$ from the perspective of the server, i.e., \[\mathbb{P}(\boldsymbol{W}= W'| Q^{[W,S,C]},X_1,\dots,X_K)= \frac{1}{K}\quad W' \in [K].\] This condition is referred to as the \emph{privacy condition}. 

Upon receiving $Q^{[W,S,C]}$, the server sends to the user an answer $A^{[W,S,C]}$, which is a (deterministic) function of the query $Q^{[W,S,C]}$ and the messages in $X$, i.e., \[H(A^{[W,S,C]}| Q^{[W,S,C]},X_1,\dots,X_K) = 0.\] The answer $A^{[W,S,C]}$ along with the side information $Y^{[S,C]}$ must enable the user to retrieve the demand $X_W$, i.e., \[H(X_W| A^{[W,S,C]}, Q^{[W,S,C]}, Y^{[S,C]})=0.\] This condition is referred to as the \emph{recoverability condition}. 

By the privacy and recoverability conditions, it follows that for any $W$, $S$, $C$ and any $W'\in [K]$, there exists $Y^{[S',C']}$ for some $S'\in \mathcal{S}$ and some $C'\in \mathcal{C}$ such that $I^{[W',S']} = I^{[W,S]}$, and \[H(X_{W'}| A^{[W,S,C]}, Q^{[W,S,C]}, Y^{[S',C']}) = 0.\] If there is no $Y^{[S',C']}$ such that $X_{W'}$ is recoverable from $A^{[W,S,C]}$ and $Y^{[S',C']}$, i.e., $H(X_{W'}| A^{[W,S,C]}, Q^{[W,S,C]}, Y^{[S',C']})\neq 0$ for all $S'\in \mathcal{S}$ and all $C'\in \mathcal{C}$ such that $I^{[W',S']} = I^{[W,S]}$, then from the server's perspective, $W'$ cannot be the user's demand index, i.e., $\mathbb{P}(\boldsymbol{W} = W'| Q^{[W,S,C]}) = 0$, and $W$ cannot be private.

For each model (I or~II), the problem is to design a query $Q^{[W,S,C]}$ and an answer $A^{[W,S,C]}$ for any $W$, $S$, and $C$ that satisfy the privacy and recoverability conditions. We refer to this problem as \emph{single-server single-message Private Information Retrieval (PIR) with Coded Side Information (CSI)}, or \emph{PIR-CSI} for short. Specifically, we refer to PIR-CSI under the model~I as the \emph{PIR-CSI--I problem}, and PIR-CSI under the model~II as the \emph{PIR-CSI--II problem}.

A collection of $Q^{[W,S,C]}$ and $A^{[W,S,C]}$ for all $W$, $S$, and $C$ such that $I^{[W,S]}=0$ or $I^{[W,S]}=1$, which satisfy the privacy and recoverability conditions, is referred to as a \emph{PIR-CSI--I protocol} or a \emph{PIR-CSI--II protocol}, respectively. 

The \emph{rate} of a PIR-CSI--I or {PIR-CSI--II} protocol is defined as the ratio of the entropy of a message, i.e., $L$, to the average entropy of the answer, i.e., $H(A^{[\boldsymbol{W},\boldsymbol{S},\boldsymbol{C}]})$, where the average is taken over all $W$, $S$, and $C$ such that $I^{[W,S]} = 0$ or $I^{[W,S]} = 1$, respectively. That is, for a PIR-CSI--I or PIR-CSI--II protocol, $H(A^{[\boldsymbol{W},\boldsymbol{S},\boldsymbol{C}]})$ is given by
\[\sum H(A^{[W,S,C]})p_{\boldsymbol{W}|\boldsymbol{S}}(W|S)p_{\boldsymbol{S}}(S)p_{\boldsymbol{C}}(C),\] where the summation is over all $W$, $S$, and $C$ such that $I^{[W,S]} = 0$ or $I^{[W,S]} = 1$, respectively. 

The \emph{capacity} of PIR-CSI--I or PIR-CSI--II problem, respectively denoted by $C_{\text{\it I}}$ or $C_{\text{\it II}}$, is defined as the supremum of rates over all PIR-CSI--I or PIR-CSI--II protocols, respectively. (The notations $C_{\text{\it I}}$ and $C_{\text{\it II}}$ should not be confused with the notation for set $C$.)

In this work, our goal is to characterize $C_{\text{\it I}}$ and $C_{\text{\it II}}$, and to design a PIR-CSI--I protocol that achieves the capacity $C_{\text{\it I}}$ and a PIR-CSI--II protocol that achieves the capacity $C_{\text{\it II}}$. 

\section{Main Results}
In this section, we present our main results. Theorem~\ref{thm:PIRCSI-I} characterizes the capacity of PIR-CSI--I problem, $C_{\text{\it I}}$, and Theorem~\ref{thm:PIRCSI-II} characterizes the capacity of PIR-CSI--II problem, $C_{\text{\it II}}$, for different values of $K$ and $M$. The proofs of Theorems~\ref{thm:PIRCSI-I} and~\ref{thm:PIRCSI-II} are given in Sections~\ref{sec:PIRCSI-I} and~\ref{sec:PIRCSI-II}, respectively. 

\begin{theorem}\label{thm:PIRCSI-I}
The capacity of PIR-CSI--I problem with $K$ messages and side information size $0\leq M< K$ is given by
\[
C_{\text{\it I}}={\left\lceil \frac{K}{M+1} \right\rceil}^{-1}.
\]
\end{theorem}

The proof consists of two parts. In the first part, we lower bound the average entropy of the answer, $H(A^{[\boldsymbol{W},\boldsymbol{S},\boldsymbol{C}]})$, or equivalently, upper bound the rate of any PIR-CSI--I protocol. In the second part, we construct a {PIR-CSI--I} protocol which achieves this rate upper-bound. 

\begin{theorem}\label{thm:PIRCSI-II}
The capacity of PIR-CSI--II problem with $K$ messages and side information size $0< M\leq K$ is given by
\begin{equation*}
C_{\text{\it II}}=
\begin{cases}
\infty, &\quad M = 1,\\
1, & \quad M=2,K\\
\frac{1}{2}, & \quad 3 \leq M \leq K-1.
\end{cases}
\end{equation*}
\end{theorem}

For the case of $M=1$, the proof is straightforward. In this case, the user has one (and only one) message in its side information, and it demands the same message. A PIR-CSI--II protocol for this case is to send no query, and receive no answer. Since the (average) entropy of the answer is zero, the rate of this protocol is infinity, and so is the capacity. 

For each of the other cases of $M$, the proof consists of two parts. In the first part, we provide a lower bound on $H(A^{[\boldsymbol{W},\boldsymbol{S},\boldsymbol{C}]})$, or  equivalently, an upper bound on the rate of any PIR-CSI--II protocol, for each case. In the second part, we construct a {PIR-CSI--II} protocol for each case which achieves the corresponding upper-bound on the rate.

\section{The PIR-CSI--I Problem}\label{sec:PIRCSI-I}

\subsection{Proof of Converse for Theorem~\ref{thm:PIRCSI-I}}

\begin{lemma}\label{lem:Converse1}
For $0\leq M<K$, $C_{\text{\it I}}\leq {\lceil \frac{K}{M+1}\rceil}^{-1}$.
\end{lemma}

\begin{proof} Suppose that the user wishes to retrieve $X_W$ for a given $W \in [K]$, and it knows $Y = Y^{[W,S]}$ for given $S\in \mathcal{S}$ and $C\in \mathcal{C}$ such that ${I^{[W,S]}=0}$. The user sends to the server a query $Q = Q^{[W,S,C]}$, and the server responds to the user by an answer $A = A^{[W,S,C]}$. We need to show that $H(A^{[\boldsymbol{W},\boldsymbol{S},\boldsymbol{C}]})$ is lower bounded by $\lceil \frac{K}{M+1}\rceil L$. Since $H(A^{[\boldsymbol{W},\boldsymbol{S},\boldsymbol{C}]})$ is the average entropy of the answer, it suffices to show that $H(A)$ is lower bounded by $\lceil \frac{K}{M+1}\rceil L$. The proof proceeds as follows: 
\begin{align}
H(A) &\geq H(A|Q,Y)\nonumber \\ &=H(A,X_W|Q,Y)\label{eq:line1}\\
&=H(X_W|Q,Y) + H(A|Q,X_W,Y) \nonumber\\
&=H(X_W)+H(A|Q,X_W,Y)\label{eq:line2}
\end{align} where~\eqref{eq:line1} holds because $H(A,X_W|Q,Y) = H(A|Q,Y)+H(X_W|A,Q,Y)$, and $H(X_W|A,Q,Y)=0$ (by the recoverability condition); and~\eqref{eq:line2} holds since $X_W$ is independent of $(Q,Y)$ (noting that $W\not\in S$), and $H(X_W|Q,Y) = H(X_W)$. 

Now, we lower bound $H(A|Q,X_W,Y)$. There are two cases: (i) $W \cup S = [K]$, and (ii) $W \cup S\neq [K]$. In the case (i), $M = K-1$, and so, $\lceil \frac{K}{M+1}\rceil L = L$. Since $H(A|Q,X_W,Y)\geq 0$, then $H(A)\geq H(X_W) = L$ (by~\eqref{eq:line2}). 

In the case (ii), we arbitrarily choose a message, say $X_{W_1}$, from the set of remaining messages, i.e., ${W_1 \not\in W \cup S}$. By the privacy and recoverability conditions, there exists $Y_1 = Y^{[S_1,C_1]}$ for some $S_1\in \mathcal{S}$ and some ${C_1\in \mathcal{C}}$ such that $I^{[W_1,S_1]} = 0$ and $H(X_{W_1}|A,Q,Y_1) = 0$. Since conditioning does not increase the entropy, then $H(X_{W_1}|A,Q,X_W,Y,Y_1) = 0$. Thus, \begin{align}
H(A|Q,X_W,Y) &\geq H\big(A|Q,X_W,Y,Y_1\big)\nonumber\\
&= H\big(A,X_{W_1}|Q,X_W,Y,Y_1\big)\label{eq:line3}\\
&=H(X_{W_1}|Q,X_W,Y,Y_1)\nonumber \\ 
& \quad +H(A|Q,X_W,X_{W_1},Y,Y_1)\nonumber\\ & = H(X_{W_1})\nonumber \\ & \quad +H(A|Q,X_W,X_{W_1},Y,Y_1)\label{eq:line4}
\end{align} where~\eqref{eq:line3} holds because $H(A,X_{W_1}|Q,X_W,Y,Y_1) = H(A|Q,X_W,Y,Y_1)+H(X_{W_1}|A,Q,X_W,Y,Y_1)$, and $H(X_{W_1}|A,Q,X_W,Y,Y_1)=0$ (by the assumption); and~\eqref{eq:line4} follows from the independence of $X_{W_1}$ and $(Q,X_W,Y,Y_1)$ (noting that $W_1\not\in W\cup S\cup S_1$), and $H(X_{W_1}|Q,X_W,Y,Y_1) = H(X_{W_1})$. 

Let $n \triangleq  \lceil \frac{K}{M+1}\rceil$. Similarly as above, it can be shown that for all $i\in [n-1]$ there exist $W_1,\dots,W_{i}\in [K]$, ${S_1,\dots,S_{i}\in \mathcal{S}}$, and $C_{1},\dots,C_{i}\in \mathcal{C}$ (and accordingly, $Y_1,\dots,Y_{i}$), where ${W_i\not\in \cup_{j\in [i-1]} (W_{j}\cup S_{j}) \cup (W\cup S)}$, such that $I^{[W_i,S_i]} = 0$, and \[{H(X_{W_{i}}|A,Q,X_W,X_{W_1},\dots,X_{W_{i-1}},Y,Y_1,\dots,Y_{i})=0}.\] Note that $|\cup_{j\in [i]} (W_j \cup S_j) \cup (W\cup S)|\leq (M+1)(i+1)$ for all $i\in [n-1]$. Repeating a similar argument as before, 
\begin{align*}
& H(A|Q,X_{W},X_{W_1},\dots,X_{W_{i-1}},Y,Y_1,\dots,Y_{i-1}) \geq H(X_{W_{i}}) \nonumber \\ &\quad +H(A|Q,X_{W},X_{W_1},\dots,X_{W_{i}},Y,Y_1,\dots,Y_{i})
\end{align*} for all $i\in [n-1]$. Putting these lower bounds together,
\begin{dmath*}
H(A|Q,X_W,Y) \geq \sum_{i=1}^{n-1} H(X_{W_i})+ H(A|Q,X_W,X_{W_1},\dots,X_{W_{n-1}},Y,Y_1,\dots,Y_{n-1}),	
\end{dmath*} and subsequently, 
\begin{dmath}\label{eq:line9}
H(A|Q,X_W,Y) \geq \sum_{i=1}^{n-1} H(X_{W_i}) = (n-1) H(X_W)	
\end{dmath} since $H(X_{W_1})=\dots=H(X_{W_{n-1}})=H(X_W)$. Putting~\eqref{eq:line2} and~\eqref{eq:line9} together, 
\begin{equation*}
H(A)\geq n H(X_W) = \left\lceil\frac{K}{M+1}\right\rceil L,
\end{equation*} as was to be shown.
\end{proof}

\subsection{Proof of Achievability for Theorem \ref{thm:PIRCSI-I}}\label{subsec:AchThm1}
In this section, we propose a PIR-CSI--I protocol for arbitrary $K\geq 1$ and ${0\leq M\leq K-1}$. 

Assume, without loss of generality (w.l.o.g.), that $S = \{1,\dots,M\}$ and $C = \{c_{1},\dots,c_{M}\}$. 

\subsubsection*{Randomized Partitioning (RP) Protocol} The RP protocol consists of four steps as follows: 

\textbf{\it Step 1:} The user constructs $n\triangleq\lceil \frac{K}{M+1}\rceil$ (ordered) sets $Q_1,\dots,Q_n$ of indices in $[K]$, each of size $M+1$, and $n$ (ordered) sets $Q'_1,\dots,Q'_n$ of elements in $\mathbb{F}^{\times}_q$, each of size $M+1$. 

For constructing $Q_1,\dots,Q_n$, ${l \triangleq (M+1)n-K}$ extra indices are required. The procedure of selecting these extra indices is as follows. First, the user randomly chooses two integers $s$ and $r$ according to a joint pmf $p_{\boldsymbol{s},\boldsymbol{r}}(s,r)$ given by 
\begin{equation*}
p_{\boldsymbol{s},\boldsymbol{r}}(s,r) = 
\begin{cases}
\alpha_{n,r}\beta_{s,r}P, & s+r = l-1,\\
2\alpha_{n,r}\beta_{s,r}P, & s+r = l,
\end{cases}
\end{equation*} where $\alpha_{n,r}=1$ for $1\leq n\leq 2$ and $0\leq r\leq l$, and 
\[
\alpha_{n,r}=\frac{((M+1)(n-1)-2r)! ((M+1)!)^2}{((M+1)(n-1))! ((M-r+1)!)^2}
\] 
for $n\geq 3$ and $0\leq r\leq l$; ${\beta_{s,r}={\binom{M}{s}\binom{K-M-1}{r}}/{\binom{M}{l-1}}}$ for all $s$ and $r$ such that ${l-1\leq s+r\leq l}$, and $\beta_{s,r}=0$ otherwise; and $P$ is the (unique) solution of the equation ${\sum_{s,r} p_{\boldsymbol{s},\boldsymbol{r}}(s,r)=1}$ where the sum is over all $s$ and $r$. 

If $s+r = l$, the user randomly selects $s$ indices from $S$ and $r$ indices from $R\triangleq [K]\setminus (W\cup S)$. If $s+r=l-1$, the user selects the index $W$ along with $s$ and $r$ randomly chosen indices from $S$ and $R$, respectively. Denote by $V$ the set of $r$ selected indices from $R$, and by $U$ the set of $l$ selected indices from $W$, $S$, and $R$. Note that the probability of any specific realization of $U$ is given by
\begin{equation*}
\frac{p_{\boldsymbol{s},\boldsymbol{r}}(s,r)}{\binom{M}{s}\binom{K-M-1}{r}}	.
\end{equation*}

%


%


Next, the user creates the set $Q_1 = \{W,1,\dots,M\}$, and assigns all indices in $V$ to the set $Q_2$ (if exists, i.e., $n\geq 2$) and the set $Q_3$ (if exists, i.e., $n\geq 3$). Then, the user assigns $M+1-r$ randomly selected indices from ${U\cup R\setminus V}$ (or respectively, $U\cup R\setminus Q_2$) to $Q_2$ (or respectively, $Q_3$). Next, the user randomly partitions all ${(M+1)(n-1)-2r}$ indices in ${U\cup R\setminus (Q_2\cup Q_3)}$ (if any) into the remaining $n-3$ sets $Q_4,\dots,Q_n$ (if exist, i.e., $n\geq 4$), each of size $M+1$. Note that the probability of a specific realization of such a partitioning is given by
\begin{equation*}
P_{n,r} \triangleq
\begin{cases}
\frac{2 (n-3)!((M-r+1)!)^2 ((M+1)!)^{n-3}}{((M+1)(n-1)-2r)!}, & \hspace{-0.1cm} n\geq 3, 0\leq r\leq l,\\
1,& \hspace{-0.1cm} n <3, 0\leq r\leq l.
\end{cases} 
\end{equation*}

For constructing $Q'_1,\dots, Q'_n$, the user creates the set $Q'_1=\{c,c_{1},\dots,c_{M}\}$ where $c$ is chosen from $\mathbb{F}^{\times}_q$ at random, and it creates each of the sets $Q'_2,\dots,Q'_n$ by randomly choosing $M+1$ elements from $\mathbb{F}^{\times}_q$. 

\textbf{\it Step 2:} The user randomly rearranges the elements of each set $Q_i$ and $Q'_i$, and constructs $Q^{*}_i = (Q_i,Q'_i)$ for all $i\in [n]$. The user then reorders $Q^{*}_1,\dots,Q^{*}_n$ by a randomly chosen permutation $\sigma: [n]\mapsto [n]$, and sends to the server the query $Q^{[W,S,C]} = \{Q^{*}_{\sigma^{-1}(1)},\dots,Q^{*}_{\sigma^{-1}(n)}\}$.

\textbf{\it Step 3:} By using $Q^{*}_i=(Q_i,Q'_i)$, the server computes $A_{i} = \sum_{j=1}^{M+1} c_{i_j} X_{i_j}$ for all $i\in [n]$ where $Q_{i} = \{i_1,\dots,i_{M+1}\}$ and $Q'_{i} = \{c_{i_1},\dots,c_{i_{M+1}}\}$, and it sends to the user the answer $A^{[W,S,C]}=\{A_{\sigma^{-1}(1)},\dots,A_{\sigma^{-1}(n)}\}$.

\textbf{\it Step 4:} Upon receiving the answer from the server, the user retrieves $X_W$ by subtracting off the contribution of its side information $Y^{[S,C]}$ from $A_{\sigma(1)}=cX_W+\sum_{i=1}^{M} c_{i}X_{i}$.

\begin{lemma}\label{lem:Ach1}
The RP protocol is a PIR-CSI--I protocol, and achieves the rate ${\lceil \frac{K}{M+1} \rceil}^{-1}$.
\end{lemma}

\begin{proof}
In the RP protocol (Step~3), the answer $A^{[W,S,C]}$ consists of $n$ pieces of information $A_1,\dots,A_{n}$, where each $A_i$ is a linear combination of $M+1$ messages in $X$. Since $X_1,\dots,X_K$ are uniformly and independently distributed over $\mathbb{F}_{q^m}$ and $A_1,\dots,A_n$ are linearly independent combinations of $X_1,\dots,X_K$ over $\mathbb{F}_q$, then $A_1,\dots,A_n$ are uniformly and independently distributed over $\mathbb{F}_{q^m}$. That is, $H(A_1)=\dots=H(A_n)=m\log q=L$, and $H(A^{[W,S,C]})=H(A_1)+\dots+H(A_n)=nL$. Since $H(A^{[W,S,C]})=nL$ for all $W,S,C$ such that $I^{[W,S]}=0$, then the average entropy of the answer over all $W,S,C$ such that $I^{[W,S]}=0$, i.e., $H(A^{[\boldsymbol{W},\boldsymbol{S},\boldsymbol{C}]})$, is equal to $nL$. Thus, the rate of the RP protocol is equal to $\frac{L}{nL}=\frac{1}{n}={\lceil \frac{K}{M+1}\rceil}^{-1}$. 

From Step~4 of the RP protocol, it should be obvious that the recoverability condition is satisfied. To prove that the RP protocol satisfies the privacy condition, we need to show that $\mathbb{P}(\boldsymbol{W}=W'|Q^{[W,S,C]},X_1\dots,X_K)=\frac{1}{K}$ for all $W'\in [K]$. Since the RP protocol does not depend on the contents of the messages $X_1,\dots,X_K$, then it is sufficient to prove that $\mathbb{P}(\boldsymbol{W}=W'|Q^{[W,S,C]})=\frac{1}{K}$ for all $W'\in [K]$. By the application of the total probability theorem and Bayes' rule, to show that the RP protocol satisfies the privacy condition, it suffices to show that $\mathbb{P}(Q^{[W,S,C]}|\boldsymbol{W}=W')$ is the same for all $W'\in [K]$. Since all possible collections $\{Q'_1,\dots,Q'_n\}$ are equiprobable, it suffices to show that $\mathbb{P}(Q|\boldsymbol{W}=W')$ is the same for all $W'\in [K]$, where $Q\triangleq\{Q_1,\dots,Q_n\}$. 

Each set $Q_i$ consists of two disjoint subsets $I_i$ and ${J_i = Q_i\setminus I_i}$ where $I_i$ is the set of all indices in $Q_i$ that belong to no other set $Q_j$. Note that $|I_1 \cup \cdots \cup I_n|=K-l$ and $|J_1 \cup \cdots \cup J_n|=l$. Consider an arbitrary $W'\in [K]$. There are two cases: (i) $W'\in J_i$ and $W'\in J_{j}$ for some $i,j\in [n]$, $i\neq j$, and (ii) $W'\in I_k$ for some $k\in [n]$. 

In the case (i), 
\begin{align*}
\mathbb{P}(Q|\boldsymbol{W}=W') &= \mathbb{P}(Q| \boldsymbol{W}=W' , \boldsymbol{S}=I_i \cup J_i \setminus \{W'\}) \\ & \quad \quad \times \mathbb{P}(\boldsymbol{S}=I_i \cup J_i\setminus \{W'\}|\boldsymbol{W}=W')\\ &\quad + \mathbb{P}(Q| \boldsymbol{W}=W' , \boldsymbol{S}=I_{j} \cup J_{j}\setminus \{W'\}) \\ & \quad\quad \times \mathbb{P}(\boldsymbol{S}=I_{j} \cup J_{j}\setminus \{W'\}|\boldsymbol{W}=W')\\ &= \frac{p_{\boldsymbol{s},\boldsymbol{r}}(s_i,r_i)}{\binom{M}{s_i}\binom{K-M-1}{r_i}\binom{K-1}{M}}P_{n,r_i}\\ &\quad +\frac{p_{\boldsymbol{s},\boldsymbol{r}}(s_{j},r_{j})}{\binom{M}{s_{j}}\binom{K-M-1}{r_{j}}\binom{K-1}{M}}P_{n,r_{j}}
\end{align*} where $s_i=|J_i|-1$, $r_i = l-|J_i|$, $s_{j}=|J_{j}|-1$, and ${r_{j}=l-|J_{j}|}$, and $p_{\boldsymbol{s},\boldsymbol{r}}(s,r)$ and $P_{n,r}$ are defined as in the protocol. Note that $\mathbb{P}(\boldsymbol{S}=I_i\cup J_i\setminus\{W'\}|\boldsymbol{W}=W')$ and $\mathbb{P}(\boldsymbol{S}=I_i\cup J_i\setminus\{W'\}|\boldsymbol{W}=W')$ are equal to $\frac{1}{\binom{K-1}{M}}$. (Note that ${0 \leq s_i,r_i,s_{j},r_{j} \leq l-1}$.) In the case (ii), 
\begin{align*}
\mathbb{P}(Q|\boldsymbol{W}=W') & = \mathbb{P}(Q|\boldsymbol{W}=W', \boldsymbol{S}=I_k \cup J_k\setminus \{W'\}) \\ &\quad\quad \times \mathbb{P}(\boldsymbol{S}=I_k \cup J_k\setminus \{W'\}|\boldsymbol{W}=W')\\ &= \frac{p_{\boldsymbol{s},\boldsymbol{r}}(s_k,r_k)}{\binom{M}{s_k}\binom{K-M-1}{r_k}\binom{K-1}{M}}P_{n,r_k}
\end{align*} where $s_k = |J_k|$, and $r_k = l-|J_k|$. (Note that ${0\leq s_k,r_k\leq l}$.) Define \[f(s,r) \triangleq \frac{p_{\boldsymbol{s},\boldsymbol{r}}(s,r)}{\binom{M}{s}\binom{K-M-1}{r}\binom{K-1}{M}}P_{n,r}, \quad 0\leq s,r\leq l.\]  Note that, in the case (i), $\mathbb{P}(Q|\boldsymbol{W}=W')=f(s_i,r_i)+f(s_{j},r_{j})$ for some $s_i,r_i,s_{j},r_{j}$ such that $s_i+r_i = l-1$ and $s_{j}+r_{j} = l-1$, and in the case (ii), $\mathbb{P}(Q|\boldsymbol{W}=W')=f(s_k,r_k)$ for some $s_k,r_k$ such that $s_k+r_k = l$. Thus, it should not be hard to see that the privacy condition is met so long as the following equations hold: 
\begin{equation}\label{eq:line6}
f(s,r)+ f(s',r') = f(s'',r'')+f(s''',r''') 
\end{equation} for all $0\leq s,r,s',r',s'',r'',s''',r'''\leq l-1$ such that $s+r=s'+r'=s''+r''=s'''+r'''=l-1$;
\begin{equation}\label{eq:line7}
f(s,r)+ f(s',r') = f(s'',r'')
\end{equation} for all $0\leq s,r,s',r'\leq l-1$ and all $0\leq s'',r''\leq l$ such that $s+r=s'+r'=l-1$ and $s''+r''=l$, and 
\begin{equation}\label{eq:line8}
f(s,r) = f(s',r')
\end{equation} for all $0\leq s,r,s',r'\leq l$ such that $s+r=s'+r'=l$. By a simple algebra, one can verify that for the choice of $p_{\boldsymbol{s},\boldsymbol{r}}(s,r)$ specified in the protocol, the equations~\eqref{eq:line6}-\eqref{eq:line8} are met, and the RP protocol satisfies the privacy condition. This completes the proof.
\end{proof}

\section{The PIR-CSI--II Problem}\label{sec:PIRCSI-II}

\subsection{Proof of Converse for Theorem~\ref{thm:PIRCSI-II}}

\begin{lemma}\label{lem:Converse2}
For $M=2$, $C_{\text{\it II}}\leq 1$; for $2<M\leq K-1$, $C_{\text{\it II}}\leq \frac{1}{2}$, and for $M=K$, $C_{\text{\it II}}\leq 1$.
\end{lemma}

\begin{proof} 
Fix an arbitrary $W\in [K]$ and $Y=Y^{[S,C]}$ for arbitrary $S\in \mathcal{S}$ and $C\in \mathcal{C}$ such that $I^{[W,S]}=1$. Let $S = \{i_1,\dots,i_M\}$ and $C = \{c_{i_1},\dots,c_{i_M}\}$. Consider a query $Q=Q^{[W,S,C]}$ and an answer $A=A^{[W,S,C]}$. 

For the cases of $M=2$ and $M=K$, it suffices to show that $H(A)\geq L$. Note that $H(A)\geq H(A|Q,Y)=H(A,X_W|Q,Y)$, where the equality follows from the recoverability condition and the chain rule of entropy, and $H(A,X_W|Q,Y)= H(X_W|Q,Y)+H(A|Q,X_W,Y)\geq H(X_W)$, where the inequality follows from the independence of $X_W$ and $(Q,Y)$, and the non-negativity of entropy. Putting these arguments together, $H(A)\geq H(X_W)=L$.  

For the cases of $3\leq M \leq K-1$, we need to show that $H(A)\geq 2L$. By the argument above,
\begin{equation}\label{eq:line9}
H(A) \geq H(X_W)+H(A|Q,X_W,Y).
\end{equation} To lower bound $H(A|Q,X_W,Y)$, we arbitrarily choose a message, say $X_{W'}$, such that ${I^{[W',S]}=1}$. By the privacy and recoverability conditions, there exists $Y' = Y^{[S',C']}$ for some $S'\in \mathcal{S}$ and some $C'\in \mathcal{C}$ such that $I^{[W',S']} = 1$ and $H(X_{W'}|A,Q,Y') = 0$. Since conditioning does not increase the entropy, then $H(X_{W'}|A,Q,X_W,Y,Y') = 0$. Thus, 
\begin{align}
H(A|Q,X_W,Y) & \geq H(A|Q, X_W,Y,Y') \nonumber \\
&=H(A,X_{W'}|Q,X_W,Y,Y') \label{eq:line10}\\
&=H(X_{W'}|Q,X_W,Y,Y')\nonumber \\&\quad + H(A|Q,X_W,X_{W'},Y,Y'), \label{eq:line11}
\end{align} where~\eqref{eq:line10} holds because $H(A,X_{W'}|Q,X_W,Y,Y') = H(A|Q,X_W,Y,Y')+H(X_{W'}|A,Q,X_W,Y,Y')$, and $H(X_{W'}|A,Q,X_W,Y,Y')=0$ (by the assumption); and~\eqref{eq:line11} follows from the chain rule of entropy. 

Since $X_W$, $Y$, $X_{W'}$, and $Y'$ are linear functions, either (i) $H(X_{W'}|X_W,Y,Y')=H(X_{W'})$, i.e., $X_{W'}$ is independent of $(X_W,Y,Y')$, or (ii) ${H(X_{W'}|X_W,Y,Y')=0}$, i.e., $X_{W'}$ is recoverable from $(X_W,Y,Y')$. In the case (i), $H(X_{W'}|X_W,Y,Y')=H(X_{W'})$, and accordingly, $H(X_{W'}|Q,X_W,Y,Y')=H(X_{W'})$ since ${Q\rightarrow (X_W,Y,Y')\rightarrow X_{W'}}$ is a Markov chain. 

Rewriting~\eqref{eq:line11}, 
\begin{align}
H(A|Q,X_W,Y) &\geq H(X_{W'})\nonumber\\ &\quad +H(A|Q,X_W,X_{W'},Y,Y')\nonumber\\
&\geq H(X_{W'}).\label{eq:line12} 
\end{align} By~\eqref{eq:line9} and~\eqref{eq:line12}, $H(A)\geq H(X_W)+H(X_{W'})=2L$.

In the case (ii), $H(X_{W'}|X_W,Y,Y')=0$, and subsequently, $H(X_{W'}|Q,X_W,Y,Y')=0$. Again, by the linearity of $X_W$, $Y$, $X_{W'}$, and $Y'$, it follows that \[Y=c_W X_W+c_{W'} X_{W'}+Z\] and \[Y'=c'_W X_W + c'_{W'} X_{W'}+c' Z\] for some $c'_W,c'_{W'},c'\in \mathbb{F}^{\times}_q$, where $Z = \sum_{i\in S\setminus \{W,W'\}} c_{i} X_i$. 

To lower bound $H(A|Q,X_W,Y)$, we choose an arbitrary message, say $X_{W''}$, such that $I^{[W'',S]}=0$. Again, by the privacy and recoverability conditions, there exists ${Y'' = Y^{[S'',C'']}}$ for some $S''\in \mathcal{S}$ and some $C''\in \mathcal{C}$ such that $I^{[W'',S'']} = 1$ and $H(X_{W''}|A,Q,Y'') = 0$, and accordingly, $H(X_{W''}|A,Q,X_W,Y,Y'') = 0$. Similar to~\eqref{eq:line11}, 
\begin{align}
H(A|Q,X_W,Y) & \geq H(X_{W''}|Q,X_W,Y,Y'')\nonumber \\&\quad + H(A|Q,X_W,X_{W''},Y,Y''). \label{eq:line13}
\end{align} Similarly as in the case (i), if $X_{W''}$ is independent of $(Q,X_W,Y,Y'')$, then $H(A)\geq H(X_W)+H(X_{W''})=2L$. If $X_{W''}$ is recoverable from $(Q,X_W,Y,Y'')$, then \[{Y''=c''_{W''} X_{W''}+c'' (c_{W'} X_{W'}+Z)}\] for some $c''_{W''},c''\in \mathbb{F}^{\times}_q$. Note that $X_{W''}$ is independent of $(Q,X_{W'},Y',Y'')$ since $X_{W''}$ is not recoverable from $c'_{W}X_W+c'Z$ and $c''_{W''}X_{W''}+c''Z$, or equivalently, from $Y'$ and $Y''$ given $X_{W'}$. Also, $X_{W'}$ is independent of $(Q,Y',Y'')$ since $X_{W'}$ is not recoverable from $Y'$ and $Y''$. Thus,
\begin{align}
H(A) &\geq H(A|Q,Y',Y'')\nonumber \\ &=H(A,X_{W'},X_{W''}|Q,Y',Y'') \label{eq:line14}\\
&=H(X_{W'},X_{W''}|Q,Y',Y'')\nonumber\\
&\quad +H(A|Q,X_{W'},X_{W''},Y',Y'')\nonumber\\
&\geq H(X_{W'},X_{W''}|Q,Y',Y'')\nonumber\\
& \geq H(X_{W'}|Q,Y',Y'')\nonumber\\
&\quad +H(X_{W''}|Q,X_{W'},Y',Y'')\nonumber\\
& = H(X_{W'})+H(X_{W''})\label{eq:line15}
\end{align} where~\eqref{eq:line14} holds because $H(X_{W'}|A,Q,Y')=0$ and $H(X_{W''}|A,Q,Y'')=0$ (by the recoverability condition), and accordingly, $H(X_{W'}|A,Q,Y',Y'')=0$ and $H(X_{W''}|A,Q,X_{W'},Y',Y'')=0$;~\eqref{eq:line15} holds since $X_{W'}$ and $X_{W''}$ are independent of $(Q,Y',Y'')$ and $(Q,X_{W'},Y',Y'')$, respectively. Thus, $H(A)\geq H(X_{W'})+H(X_{W''})=2L$.
\end{proof}

\subsection{Proof of Achievability for Theorem \ref{thm:PIRCSI-II}}\label{subsec:AchThm2}
In this section, we propose a PIR-CSI--II protocol for each of the cases of $M=2$ (case 1), $3\leq M\leq \frac{K}{2}$ (case 2), $\frac{K}{2}+1\leq M\leq K-1$ (case 3), and $M=K$ (case 4). 

Assume, w.l.o.g., that $W=\{1\}$, ${S = \{1,\dots,{M}\}}$, and ${C = \{c_{1},\dots,c_{M}\}}$. 

\textit{Proposed Protocols for Cases 1-4:} The proposed protocol for each case consists of four steps, where the steps 2-4 are the same as the steps 2-4 in the RP protocol (Section~\ref{subsec:AchThm1}), and the step~1 of the proposed protocols are as follows: 

\textit{Case 1:} The user randomly selects the demand index $W$, (i.e., $1$) with probability $\frac{1}{K}$, or the other index in $S$ (i.e., $2$) with probability $\frac{K-1}{K}$, and it creates two sets $Q_1 = \{i\}$ and $Q'_1=\{c\}$, where $i$ is the selected index by the user, and $c$ is chosen uniformly at random from $\mathbb{F}^{\times}_q$.

\textit{Case 2:} The user creates two (ordered) sets $Q_1,Q_2$ of indices in $[K]$, each of size $M-1$, and two (ordered) sets $Q'_1,Q'_2$ of elements in $\mathbb{F}^{\times}_q$, each of size $M-1$. 

For constructing $Q_1$ and $Q_2$, the user first chooses an integer $r$ randomly according to a pmf $p_{\boldsymbol{r}}(r)$ given by 
\begin{equation*}
p_{\boldsymbol{r}}(r) = 
\begin{cases}
\frac{2(M-1)}{K}, & r = M-2,\\
1-\frac{2(M-1)}{K}, & r = M-1.
\end{cases}
\end{equation*} If $r = M-1$, the user randomly selects $r$ indices from $R\triangleq [K]\setminus S$. If $r=M-2$, the user selects the index $W$, and $r$ randomly chosen indices from $R$. Denote by $U$ the set of $M-1$ selected indices from $W$ and $R$. Then, the user creates the sets $Q_1 = \{2,\dots,M\}$ and $Q_2=U$.

For constructing $Q'_1$ and $Q'_2$, the user creates the set ${Q'_1=\{c_{2},\dots,c_{M}\}}$, and creates the set $Q'_2$ by randomly choosing $M-1$ elements from $\mathbb{F}^{\times}_q$. 

\textit{Case 3:} The user creates two (ordered) sets $Q_1,Q_2$ of indices in $[K]$, each of size $M$, and two (ordered) sets $Q'_1,Q'_2$ of elements in $\mathbb{F}^{\times}_q$, each of size $M$. 

For constructing $Q_1$ and $Q_2$, the user begins with choosing an integer $s$ at random according to a pmf $p_{\boldsymbol{s}}(s)$ given by 
\begin{equation*}
p_{\boldsymbol{s}}(s) = 
\begin{cases}
1-\frac{2(K-M)}{K}, & s = 2M-K-1,\\
\frac{2(K-M)}{K}, & s = 2M-K.
\end{cases}
\end{equation*} If $s = 2M-K$, the user randomly selects $s$ indices from $S\setminus W$. If $s=2M-K-1$, the user selects the index $W$ together with $s$ randomly chosen indices from $S\setminus W$. Denote by $U$ the set of $2M-K$ selected indices from $S$. Then, the user creates the sets $Q_1 = S$ and $Q_2=U\cup ([K]\setminus S)$.

The user also creates the set ${Q'_1=\{c,c_{2},\dots,c_{M}\}}$ where $c$ is randomly chosen from ${\mathbb{F}^{\times}_q\setminus \{c_{1}\}}$, and creates the set $Q'_2$ by randomly choosing $M$ elements from $\mathbb{F}^{\times}_q$. 

\textit{Case 4:} The user creates two sets $Q_1 = [K]$ and $Q'_1=\{c,c_{2},\dots,c_{K}\}$ where $c$ is randomly chosen from $\mathbb{F}^{\times}_q\setminus \{c_{1}\}$. 

\begin{lemma}\label{lem:Ach2}
The proposed protocols for $M=2$, $3\leq M\leq \frac{K}{2}$, $\frac{K}{2}+1\leq M\leq K-1$, and $M=K$ are PIR-CSI--II protocols, and achieve the rates $1$, $\frac{1}{2}$, $\frac{1}{2}$, and $1$, respectively.
\end{lemma}

\begin{proof}
The proof of the achievability rate of the proposed protocols follows the exact same line as in the proof of the achievability rate of the RP protocol in Lemma~\ref{lem:Ach1}. Moreover, by the structure of the proposed protocols, it should not be hard to see that the recoverability condition is met. Similarly as in the proof of Lemma~\ref{lem:Ach1}, to prove the privacy of the proposed protocols, it suffices to show that $\mathbb{P}(Q|\boldsymbol{W}=W')$ is the same for all $W'\in [K]$, where $Q\triangleq Q_1$ for the cases~1 and~4, and $Q\triangleq \{Q_1,Q_2\}$ for the cases~2 and~3. (Note that all possible sets $Q'_1$ are equiprobable, and all possible collections $\{Q'_1,Q'_2\}$ are equiprobable.) The proof of privacy of the proposed protocol for each case is as follows. 

\subsubsection*{Case 1} Consider an arbitrary $Q=Q_1$. Take an arbitrary $W'\in [K]$. There are two cases: (i) $Q=\{W'\}$, and (ii) $Q=\{W''\}$ for some $W''\neq W'$. 

In the case (i),
\begin{align}
\mathbb{P}(Q|\boldsymbol{W}=W')&=\hspace{-0.1cm}\sum_{W''\neq W'}\mathbb{P}(Q|\boldsymbol{W}=W', \boldsymbol{S}=\{W',W''\})\nonumber\\
&\hspace{1.7cm}\times \mathbb{P}(\boldsymbol{S}=\{W',W''\}|\boldsymbol{W}=W')\nonumber\\
& = \frac{1}{K}	\label{eq:line16}
\end{align} since 
\begin{align*}
&\sum_{W''\neq W'}\mathbb{P}(Q|\boldsymbol{W}=W', \boldsymbol{S}=\{W',W''\})\\
&\quad={(K-1)\mathbb{P}(Q|\boldsymbol{W}=W', \boldsymbol{S}=\{W',W''\})},
\end{align*} 
\[\mathbb{P}(Q|\boldsymbol{W}=W', \boldsymbol{S}=\{W',W''\})=\frac{1}{K},\] and
\[\mathbb{P}(\boldsymbol{S}=\{W',W''\}|\boldsymbol{W}=W')=\frac{1}{K-1}.\] 

In the case (ii),
\begin{align}
\mathbb{P}(Q|\boldsymbol{W}=W')&=\mathbb{P}(Q|\boldsymbol{W}=W', \boldsymbol{S}=\{W',W''\})\nonumber\\
&\quad\quad\times \mathbb{P}(\boldsymbol{S}=\{W',W''\}|\boldsymbol{W}=W')\nonumber\\
& = \frac{1}{K}	\label{eq:line17}
\end{align} since 
\[\mathbb{P}(Q|\boldsymbol{W}=W', \boldsymbol{S}=\{W',W''\})=\frac{K-1}{K},\] and \[\mathbb{P}(\boldsymbol{S}=\{W',W''\}|\boldsymbol{W}=W')=\frac{1}{K-1}.\] By~\eqref{eq:line16} and~\eqref{eq:line17}, ${\mathbb{P}(Q|\boldsymbol{W}=W')}$ is the same for all ${W'\in [K]}$.\\ 

\subsubsection*{Case 2} Consider an arbitrary $Q=\{Q_1,Q_2\}$. Take an arbitrary $W'\in [K]$. There are two cases: (i) $W'\in Q_1\cup Q_2$, and (ii) $W'\not\in Q_1\cup Q_2$. 

In the case (i), w.l.o.g., assume that $W'\in Q_1$. Thus, 
\begin{align}
\mathbb{P}(Q|\boldsymbol{W}=W')&=\mathbb{P}(Q|\boldsymbol{W}=W', \boldsymbol{S}=W'\cup Q_2)\nonumber\\
&\quad\times \mathbb{P}(\boldsymbol{S}=W'\cup Q_2|\boldsymbol{W}=W')\nonumber\\
& = \frac{2(M-1)}{K\binom{K-M}{M-2}\binom{K-1}{M-1}}\label{eq:line18}
\end{align} since 
\[\mathbb{P}(Q|\boldsymbol{W}=W',\boldsymbol{S}=W'\cup Q_2)=\frac{2(M-1)}{K\binom{K-M}{M-2}},\] and 
\[\mathbb{P}(\boldsymbol{S}=W'\cup Q_2|\boldsymbol{W}=W')=\frac{1}{\binom{K-1}{M-1}}.\] 

In the case (ii), 
\begin{align}
\mathbb{P}(Q|\boldsymbol{W}=W')&=\mathbb{P}(Q|\boldsymbol{W}=W', \boldsymbol{S}=W'\cup Q_1)\nonumber\\
&\quad\quad\times \mathbb{P}(\boldsymbol{S}=W'\cup Q_1|\boldsymbol{W}=W')\nonumber\\
&\quad+\mathbb{P}(Q|\boldsymbol{W}=W', \boldsymbol{S}=W'\cup Q_2)\nonumber\\
&\quad\quad\times \mathbb{P}(\boldsymbol{S}=W'\cup Q_2|\boldsymbol{W}=W')\nonumber\\
& = \frac{2(K-2(M-1))}{K\binom{K-M}{M-1}\binom{K-1}{M-1}}\label{eq:line19}
\end{align} since 
\begin{align*}
&\mathbb{P}(Q|\boldsymbol{W}=W',\boldsymbol{S}=W'\cup Q_1)\\ 
&\quad =\mathbb{P}(Q|\boldsymbol{W}=W',\boldsymbol{S}=W'\cup Q_2)\\ &\quad =\frac{K-2(M-1)}{K\binom{K-M}{M-1}},	
\end{align*} and
\begin{align*}
& \mathbb{P}(\boldsymbol{S}=W'\cup Q_1|\boldsymbol{W}=W')\\ 
& \quad =\mathbb{P}(\boldsymbol{S}=W'\cup Q_2|\boldsymbol{W}=W')\\
& \quad =\frac{1}{\binom{K-1}{M-1}}.	
\end{align*} Since~\eqref{eq:line18} and~\eqref{eq:line19} are equal, then $\mathbb{P}(Q|\boldsymbol{W}=W')$ is the same for all $W'\in [K]$.\\ 
 
\subsubsection*{Case 3} Consider an arbitrary $Q=\{Q_1,Q_2\}$. Take an arbitrary $W'\in [K]$. There are two cases: (i), $W'\in Q_1\cap Q_2$, and (ii) $W'\not\in Q_1\cap Q_2$. 

In the case (i), 
\begin{align}
\mathbb{P}(Q|\boldsymbol{W}=W')&=\mathbb{P}(Q|\boldsymbol{W}=W', \boldsymbol{S}= Q_1)\nonumber\\
&\quad\quad\times \mathbb{P}(\boldsymbol{S}=Q_1|\boldsymbol{W}=W')\nonumber\\
&\quad+\mathbb{P}(Q|\boldsymbol{W}=W', \boldsymbol{S}=Q_2)\nonumber\\
&\quad\quad\times \mathbb{P}(\boldsymbol{S}=Q_2|\boldsymbol{W}=W')\nonumber\\
& = \frac{2(2M-K)}{K\binom{M-1}{2M-K-1}\binom{K-1}{M-1}}\label{eq:line20}
\end{align} since 
\begin{align*}
&\mathbb{P}(Q|\boldsymbol{W}=W',\boldsymbol{S}=Q_1)\\ 
&\quad =\mathbb{P}(Q|\boldsymbol{W}=W',\boldsymbol{S}=Q_2)\\ &\quad =\frac{2M-K}{K\binom{M-1}{2M-K-1}},	
\end{align*} and
\begin{align*}
& \mathbb{P}(\boldsymbol{S}=Q_1|\boldsymbol{W}=W')\\ 
& \quad =\mathbb{P}(\boldsymbol{S}=Q_2|\boldsymbol{W}=W')\\
& \quad =\frac{1}{\binom{K-1}{M-1}}.	
\end{align*} 

In the case (ii), w.l.o.g., assume that $W'\in Q_1$. Thus, 
\begin{align}
\mathbb{P}(Q|\boldsymbol{W}=W')&=\mathbb{P}(Q|\boldsymbol{W}=W', \boldsymbol{S}=Q_1)\nonumber\\
&\quad\times \mathbb{P}(\boldsymbol{S}=Q_1|\boldsymbol{W}=W')\nonumber\\
& = \frac{2(K-M)}{K\binom{M-1}{2M-K}\binom{K-1}{M-1}}\label{eq:line21}
\end{align} since 
\[\mathbb{P}(Q|\boldsymbol{W}=W',\boldsymbol{S}=Q_1)=\frac{2(K-M)}{K\binom{M-1}{2M-K}},\] and 
\[\mathbb{P}(\boldsymbol{S}=Q_1|\boldsymbol{W}=W')=\frac{1}{\binom{K-1}{M-1}}.\] Again,~\eqref{eq:line20} and~\eqref{eq:line21} are equal, and $\mathbb{P}(Q|\boldsymbol{W}=W')$ is the same for all $W'\in [K]$.\\ 

\subsubsection*{Case 4} Since $Q=Q_1=[K]$, then $\mathbb{P}(Q|\boldsymbol{W}=W')=1$ for all $W'\in [K]$. 

This completes the proof.
\end{proof}

\bibliographystyle{IEEEtran}
\bibliography{PIR_salim,pir_bib,coding1,coding2}

\end{document}